\documentclass[12pt,a4paper]{article}

\usepackage[utf8]{inputenc}
\usepackage[T1]{fontenc}
\usepackage{amsmath,amssymb,amsthm}
\usepackage{geometry}
\usepackage{algorithm}
\usepackage{algorithmic}
\usepackage{graphicx}
\usepackage{hyperref}
\usepackage{url}
\usepackage{booktabs}
\usepackage{tabularx}
\usepackage{xcolor}
\usepackage{caption}
\usepackage{float}

\newtheorem{definition}{Definition}
\newtheorem{theorem}{Theorem}
\newtheorem{lemma}{Lemma}
\newtheorem{corollary}{Corollary}
\newtheorem{remark}{Remark}

\newcommand{\MEV}{\mathsf{MEV}}
\newcommand{\REV}{\mathsf{REV}}
\newcommand{\HKDF}{\mathsf{HKDF}}
\newcommand{\Ctx}{\mathsf{Ctx}}
\newcommand{\VDF}{\mathsf{VDF}}
\newcommand{\idcom}{\mathsf{idcom}}

\newcommand{\Order}{\mathsf{Order}}
\newcommand{\Sign}{\mathsf{Sign}}

\newcommand{\negl}{\mathsf{negl}}
\newcommand{\Adv}{\mathcal{A}}

\title{MEV-ACE: Identity-Authenticated Fair Ordering\\for Proposer-Controlled MEV Mitigation}

\author{
  \texttt{Jian Sheng Wang} \\
  \texttt{jason@acechain.io} \\
  Montreal, Canada
}
\date{November 24, 2026}

\begin{document}
    \maketitle

    \begin{abstract}
        Maximal Extractable Value (MEV) remains a structural threat to blockchain fairness because the block producer can often observe pending transactions and unilaterally decide their ordering or inclusion. Existing mitigations hide transaction contents or outsource ordering, but they typically leave two gaps unresolved: commitments are not authenticated by slashable identities, and inclusion obligations are not backed by transferable evidence that other validators can verify.

        This paper presents \textbf{MEV-ACE}, a fair-ordering protocol for \emph{proposer-controlled ordering MEV}. MEV-ACE combines three mechanisms: (1)~\emph{registered economic identities}, whose authentication keys are deterministically derived from the ACE-GF framework and bonded on-chain; (2)~\emph{authenticated commit/open messages with validator receipt thresholds}, which make admissibility and inclusion obligations independently auditable; and (3)~\emph{verifiable-delay randomness}, which determines transaction order only after the admissible commitment set is fixed.

        We formalize the protocol in a $3f+1$ validator model with threshold receipts and show three properties under standard assumptions: \emph{order-unpredictability} over the locked admissible commitment set, \emph{commitment authenticity} under EUF-CMA security of the authentication signature scheme, and \emph{accountable inclusion} for transactions that obtain threshold commit and open receipts. Under these cryptographic assumptions, timely omission-proof delivery, and economic parameters that bound one-slot deviation gains, MEV-ACE eliminates unilateral proposer discretion over the ordering and inclusion of transactions that have crossed both receipt thresholds, within every block finalized for their slot; state-conditional transactions from producer-controlled identities can still sandwich probabilistically (Section~\ref{sec:limitations}). The protocol remains single-slot in structure, requires no threshold decryption committee, and is compatible with post-quantum signatures such as ML-DSA-44.

    \end{abstract}

    % ============================================================
    \section{Introduction}
    \label{sec:intro}
    % ============================================================

    Maximal Extractable Value (MEV) refers to the profit a block producer can extract by manipulating the ordering, inclusion, or censorship of transactions within a block~\cite{daian2020flash}. Large-scale empirical studies and measurement efforts have documented substantial extracted value on Ethereum and related DeFi ecosystems~\cite{flashbots-explore,qin2022quantifying}, value that is transferred from ordinary users to searchers and cooperating validators.

    The problem is structural. In conventional blockchains, the block producer observes pending transactions in cleartext and has unilateral discretion over their ordering. This asymmetry enables three canonical attack classes:

    \begin{itemize}
        \item \textbf{Front-running}: The producer inserts a transaction before a target transaction to capture a price movement.
        \item \textbf{Sandwich attacks}: The producer brackets a target transaction with a buy-before and sell-after pair, extracting the slippage.
        \item \textbf{Censorship}: The producer selectively excludes transactions to disadvantage specific users or benefit competing transactions.
    \end{itemize}

    We collectively term these \emph{ordering-based MEV}, distinguished from \emph{information-based MEV} (e.g., cross-domain arbitrage from public oracle updates) which does not depend on the producer's ordering power.

    \subsection{Limitations of Existing Approaches}

    Prior work follows four main strategies, each with limitations.

    \paragraph{Commit-reveal schemes~\cite{canidio2023commit}.} Users hide transaction contents during an initial commitment phase and reveal later. This reduces early information leakage, but by itself does not authenticate who is allowed to create admissible commitments, nor does it provide a portable proof that an omitted commitment had to be included.

    \paragraph{Threshold encryption~\cite{shutter2023}.} Transactions are encrypted under a threshold key and decrypted only after ordering is fixed. This requires a trusted decryption committee, introduces at least one full round-trip of latency for threshold decryption, and remains vulnerable to producer--committee collusion.

    \paragraph{Fair-ordering services~\cite{chainlink-fss}.} Decentralized oracles provide ordering attestations. These introduce an external trust assumption and oracle-layer latency, and they shift, rather than remove, the requirement that some outside party honestly witness and enforce ordering.

    \paragraph{Proposer-builder separation (PBS)~\cite{flashbots-mev-boost}.} Separates block building from block proposing. This redistributes MEV rather than eliminating it: the builder still extracts ordering-based MEV, and the proposer receives a share via auction.

    \subsection{Key Observation: Authenticated Admissibility}
    \label{sec:intro:key-obs}

    Hiding transaction payloads from the producer is only part of the semantic goal of fair ordering. The goal is to ensure that, once a user has satisfied the protocol's admission rules, the producer cannot unilaterally (i) fabricate additional admissible demand, (ii) condition openings on the realized permutation, or (iii) omit the transaction without verifiable evidence. Existing mitigations often optimize the first concern (information leakage) while leaving the latter two unresolved.

    We identify the shared structural gap as \emph{unauthenticated admissibility without portable inclusion evidence}. Concretely, many designs do not require admissible commitments to be signed by registered, slashable identities, nor do they require inclusion obligations to be backed by threshold receipts that any validator can verify. This enables the following attacks:

    \begin{enumerate}
        \item \textbf{Commitment stuffing}: The producer generates many low-cost commitments to dilute the ordering pool and improve the expected placement of its own trades.
        \item \textbf{Selective non-opening}: The producer admits many candidate transactions and reveals only the subset whose realized positions are profitable.
        \item \textbf{Unprovable omission}: Honest users may have broadcast valid commitments or openings, but other validators lack a transferable proof that the producer was obligated to include them.
    \end{enumerate}

    All three attacks exploit cheap identity creation, unauthenticated admission, or the absence of globally verifiable receipts.

    \subsection{Relation to ACE-GF}
    \label{sec:intro:relation}

    MEV-ACE does not require ACE-GF at the protocol layer: any registered verification key plus slashable bond suffices. ACE-GF~\cite{acegf} is the preferred instantiation because it supplies context-isolated authentication keys from a single identity root $\REV$ via $\mathsf{Derive}(\REV, \Ctx)$, consistent with the separation of \emph{identity derivation} from \emph{authorization path} in ACE-GF. Security nevertheless rests on protocol-layer mechanisms (registered keys, receipt thresholds, delayed randomness, and bonded participation), not on local wallet ergonomics.

    \subsection{Our Contributions}

    This paper makes the following contributions:

    \begin{enumerate}
        \item \textbf{Formal model of admissible fair ordering.} We define proposer-controlled ordering MEV in a model where commitments and openings become mandatory only after they collect threshold validator receipts, and formalize three security properties: order-unpredictability, commitment authenticity, and accountable inclusion (Section~\ref{sec:model}).

        \item \textbf{The MEV-ACE protocol.} We present a single-slot protocol that combines ACE-GF-derived authentication keys, bonded identity registration, threshold receipt certificates for commit/open phases, and VDF-based delayed randomness (Section~\ref{sec:protocol}).

        \item \textbf{Security analysis.} We prove that MEV-ACE satisfies all three properties under the sequentiality of the VDF, EUF-CMA security of the authentication signature scheme, collision resistance of SHA-256, and standard BFT delivery assumptions for omission proofs. We further derive explicit economic conditions under which honest execution is the producer's best response (Section~\ref{sec:security}).

        \item \textbf{Deployment guidance.} We provide a parameterized latency and communication analysis that states the timing assumptions explicitly and keeps VDF security parameters separate from wall-clock runtime (Section~\ref{sec:performance}).

        \item \textbf{Identity-authenticated ordering architecture.} We show how ACE-GF can be used as an identity-native key-derivation layer for MEV protection without requiring seed storage or threshold decryption.
    \end{enumerate}

    \subsection{Scope and Non-Goals}

    MEV-ACE targets \emph{proposer-controlled ordering MEV} on a single chain. Information-based MEV, such as back-running public oracle price updates or cross-domain statistical arbitrage, is \emph{not} in scope. These forms of MEV exploit publicly available information and do not disappear merely because proposer ordering discretion is reduced; they require orthogonal mitigations such as batch auctions, delayed disclosure, or private execution.

    % ============================================================
    \section{Preliminaries}
    \label{sec:prelim}
    % ============================================================

    \subsection{ACE-GF Identity Framework}

    ACE-GF (Atomic Cryptographic Entity Generative Framework)~\cite{acegf} is a deterministic key-derivation framework that derives application-specific cryptographic material from an identity root called the Root Entropy Value ($\REV$). For MEV-ACE, the relevant properties are:

    \begin{itemize}
        \item \textbf{Deterministic derivation.} Application keys are derived under explicit context strings, allowing the protocol to dedicate a distinct authentication key to MEV ordering without reusing keys from unrelated applications.
        \item \textbf{Context isolation.} Keys derived under distinct contexts are computationally independent under the pseudorandomness of the underlying KDF.
        \item \textbf{Single identity root.} Honest users manage one identity root and obtain application-specific signing keys from it; protocol-level accountability is unchanged.
        \item \textbf{Post-quantum compatibility.} ACE-GF can instantiate its authentication context with ML-DSA-44 (NIST FIPS~204), allowing the ordering protocol to inherit post-quantum signature support.
    \end{itemize}

    \begin{definition}[MEV Authentication Key]
    For a user with identity root $\REV$, let $\Ctx_{\mathsf{mev}} = (\texttt{"MEV-ACE"}, \texttt{"auth"}, \mathsf{domain})$ denote the derivation context. Following the ACE-GF convention~\cite{acegf}, let
    $$
    (sk^{\mathsf{auth}}, vk^{\mathsf{auth}})
    =
    \mathsf{Derive}(\REV, \Ctx_{\mathsf{mev}})
    $$
    denote the signature key pair derived for MEV-ACE authentication via context-isolated $\HKDF$.
    \end{definition}

    \begin{definition}[Registered Identity Commitment]
    \label{def:idcom}
    For a user with authentication verification key $vk^{\mathsf{auth}}$, the public identity commitment used by MEV-ACE is
    $$
    \idcom = H(vk^{\mathsf{auth}}),
    $$
    where $H$ is SHA-256. A registered economic identity is the on-chain tuple $(\idcom, vk^{\mathsf{auth}}, D_{\mathsf{stake}})$, where $D_{\mathsf{stake}}$ is a slashable bond locked for protocol participation.
    \end{definition}

    \begin{remark}
    In MEV-ACE, protocol-level Sybil resistance does \emph{not} rely on local wallet setup cost. The enforceable cost comes from the registered bond $D_{\mathsf{stake}}$, the per-slot commitment quota $\ell$, and slashing for invalid behavior or selective non-opening. ACE-GF is used to make authenticated key derivation convenient and context-isolated, not to serve as the sole anti-Sybil mechanism.
    \end{remark}

    \subsection{Verifiable Delay Functions}

    A Verifiable Delay Function (VDF)~\cite{boneh2018vdf} is a function $f: \mathcal{X} \to \mathcal{Y}$ that requires $T$ sequential computation steps to evaluate but whose output can be verified in $O(\log T)$ time.

    \begin{definition}[VDF]
    A VDF scheme consists of three algorithms:
    \begin{itemize}
        \item $\mathsf{Setup}(1^\lambda, T) \to \mathsf{pp}$: generates public parameters for security parameter $\lambda$ and delay parameter $T$.
        \item $\mathsf{Eval}(\mathsf{pp}, x) \to (y, \pi)$: computes output $y$ with proof $\pi$ in exactly $T$ sequential steps.
        \item $\mathsf{Verify}(\mathsf{pp}, x, y, \pi) \to \{0,1\}$: verifies the output in $O(\log T)$ time.
    \end{itemize}
    \end{definition}

    We require the \emph{sequential computation assumption}: no adversary with $\mathsf{poly}(\lambda)$ parallel processors can evaluate the VDF significantly faster than $T$ sequential steps.

    \subsection{Notation}

    Throughout this paper: $H(\cdot)$ denotes SHA-256; $\|$ denotes concatenation; $[n]$ denotes the set $\{1, \ldots, n\}$; $\lambda$ is the security parameter; $\negl(\lambda)$ denotes a negligible function in $\lambda$; $\pi_\sigma$ denotes a permutation derived from seed $\sigma$; $R_i^C$ and $R_i^O$ denote threshold receipt bundles for the commit and open phases, respectively.

    % ============================================================
    \section{Formal Model}
    \label{sec:model}
    % ============================================================

    \subsection{System Model}

    We consider a blockchain with the following participants:
    \begin{itemize}
        \item \textbf{Users} $\mathcal{U} = \{u_1, \ldots, u_n\}$: each user $u_i$ controls a registered identity $(\idcom_i, vk_i^{\mathsf{auth}}, D_{\mathsf{stake}})$ and submits transactions.
        \item \textbf{Block producer} $\mathcal{P}$: the entity responsible for assembling the current block. $\mathcal{P}$ may be strategic, rational, or Byzantine, and may also control a subset of registered user identities.
        \item \textbf{Validators} $\mathcal{V} = \{v_1, \ldots, v_{3f+1}\}$: validators verify signatures, issue commit/open receipts, and reject blocks that omit transactions proven to be mandatory.
    \end{itemize}

    Time is divided into discrete \emph{slots}. Each slot has a designated block producer selected by the consensus protocol. The block producer for slot $s$ is denoted $\mathcal{P}_s$.

    A commitment or opening becomes mandatory only after it gathers a threshold receipt bundle from validators. We use $q_c$ and $q_o$ for the commit and open thresholds, respectively, and set $q_c = q_o = 2f+1$. The lower bound gives every certificate at least $f+1$ honest signers; the upper bound $n - f = 2f+1$ keeps the $f$ Byzantine validators from blocking a certificate by withholding their receipts. A commitment or opening that reaches every honest validator before the cutoff of its phase therefore obtains its certificate. One that reaches only some honest validators can be denied a certificate by the Byzantine validators, and the guarantees below do not cover it.

    \begin{definition}[Admissible Commitment]
    A commitment for slot $s$ is a tuple $(\idcom_i, c_i, s, \sigma_i^C)$. It is \emph{admissible} if:
    \begin{enumerate}
        \item $(\idcom_i, vk_i^{\mathsf{auth}}, D_{\mathsf{stake}})$ is registered and active;
        \item $\sigma_i^C$ is a valid signature under $vk_i^{\mathsf{auth}}$ on message $(\texttt{"commit"}, \idcom_i, c_i, s)$;
        \item the identity has not exceeded the per-slot quota $\ell$; and
        \item the commitment carries a validator receipt bundle $R_i^C$ containing at least $q_c$ valid validator signatures on the same message before the slot's commit cutoff.
    \end{enumerate}
    \end{definition}

    \begin{definition}[Executable Opening]
    Let $(\idcom_i, c_i, s, \sigma_i^C)$ be an admissible commitment. An opening $(tx_i, r_i, \idcom_i, s)$ is \emph{executable} if:
    \begin{enumerate}
        \item $c_i = H(tx_i \| r_i \| \idcom_i \| s)$;
        \item it carries an opening receipt bundle $R_i^O$ with at least $q_o$ valid validator signatures on message $(\texttt{"open"}, \idcom_i, c_i, s)$ before the opening cutoff; and
        \item the corresponding identity remains active and slashable.
    \end{enumerate}
    \end{definition}

    \subsection{Ordering-Based MEV: Formal Definition}

    \begin{definition}[Transaction Ordering Function]
    An \emph{ordering function} $\Order$ maps a set of transactions $\mathcal{T} = \{tx_1, \ldots, tx_m\}$ to a permutation $\sigma \in S_m$ determining their execution sequence within a block.
    \end{definition}

    \begin{definition}[Ordering-Based MEV]
    \label{def:mev}
    Let $\sigma^*$ denote the execution order chosen by a strategic producer $\mathcal{P}$ over the executable opening set of a slot, and let $\sigma_0$ denote the protocol-prescribed order. The \emph{ordering-based MEV} extracted by $\mathcal{P}$ is:
    $$\MEV(\mathcal{P}, \mathcal{T}) = \mathsf{Profit}(\mathcal{P}, \sigma^*) - \mathsf{Profit}(\mathcal{P}, \sigma_0)$$
    where $\mathsf{Profit}(\mathcal{P}, \sigma)$ is the net economic gain to $\mathcal{P}$ from executing $\mathcal{T}$ in order $\sigma$, including any producer-controlled identities and side transactions.
    \end{definition}

    \subsection{Security Properties}

    We define three properties that collectively eliminate unilateral proposer discretion over admitted transactions:

    \begin{definition}[Order-Unpredictability]
    \label{def:order-unpredict}
    A fair-ordering protocol satisfies \emph{order-unpredictability} if, for any PPT adversary $\Adv$ controlling the block producer and any target admissible commitment $c_i \in \mathcal{C}_s^*$, the probability that $\Adv$ predicts the index $\sigma_s(i)$ assigned to $c_i$ in the slot permutation, given all information available strictly before the root lock (Section~\ref{sec:order-phase}) fixes $\mathcal{C}_s^*$, satisfies:
    $$\Pr\bigl[\Adv \text{ predicts } \sigma_s(i)\bigr] \leq \frac{1}{m} + \negl(\lambda)$$
    where $m = |\mathcal{C}_s^*|$ is the number of admissible commitments in the slot. Opening occurs only after $\sigma_s$ is published; unpredictability is therefore stated over commitment indices, not over post-hoc executable openings.
    \end{definition}

    \begin{definition}[Commitment Authenticity]
    \label{def:commit-unforge}
    A fair-ordering protocol satisfies \emph{commitment authenticity} if, for any PPT adversary $\Adv$, the probability that $\Adv$ produces an admissible commitment for a registered identity $\idcom^*$ without possessing the corresponding authentication signing key is:
    \begin{align*}
    \Pr\bigl[
    \Adv \text{ outputs an admissible commitment for } \idcom^* \\
    \text{ without } sk^{\mathsf{auth}}_*
    \bigr]
    \leq \negl(\lambda)
    \end{align*}
    \end{definition}

    \begin{definition}[Accountable Inclusion]
    \label{def:censor-resist}
    A fair-ordering protocol satisfies \emph{accountable inclusion} if any transaction that has both an admissible commitment and an executable opening must appear in every valid finalized block for that slot, and any omission can be proven by a publicly verifiable omission proof derived from $(R_i^C, R_i^O)$.
    \end{definition}

    \begin{definition}[Elimination of Proposer-Controlled Ordering MEV]
    A protocol \emph{eliminates proposer-controlled ordering MEV on admitted transactions} if it simultaneously satisfies order-unpredictability, commitment authenticity, and accountable inclusion.
    \end{definition}

    % ============================================================
    \section{The MEV-ACE Protocol}
    \label{sec:protocol}
    % ============================================================

    MEV-ACE operates in four logical steps: \textsc{Register}, \textsc{Commit}, \textsc{Order}, and \textsc{Open}. Registration occurs off the slot critical path; the latter three phases execute within one slot.

    \subsection{Protocol Overview}

    The protocol combines three ideas:
    \begin{enumerate}
        \item \textbf{Authenticated admissibility}: a commitment must be signed by a registered identity and acknowledged by a threshold of validators before it enters the admissible set.
        \item \textbf{Delayed randomness}: the execution permutation is derived only after the admissible set is locked, removing the producer's ability to pick positions after learning transaction contents.
        \item \textbf{Receipt-backed inclusion}: both commitments and openings carry portable receipt bundles, so any validator can verify a claimed omission.
    \end{enumerate}

    As a result, the producer cannot cheaply fabricate additional admissible demand, nor can it silently omit a transaction that has already crossed the protocol's admission thresholds.

    \subsection{Registration}

    Before participating, user $u_i$ derives an authentication key pair
    $$
    (sk_i^{\mathsf{auth}}, vk_i^{\mathsf{auth}})
    =
    \mathsf{Derive}(\REV_i, \Ctx_{\mathsf{mev}}),
    $$
    computes $\idcom_i = H(vk_i^{\mathsf{auth}})$, and registers $(\idcom_i, vk_i^{\mathsf{auth}}, D_{\mathsf{stake}})$ on-chain. Registration activates the identity only after the bond is locked, and the bond remains slashable for non-opening and invalid protocol behavior.

    Each active identity may submit at most $\ell$ commitments per slot. This quota is protocol-enforced and is central to the anti-stuffing analysis in Section~\ref{sec:sybil-cost}.

    \subsection{Phase 1: Commit}
    \label{sec:commit-phase}

    Each user $u_i$ wishing to include a transaction $tx_i$ in slot $s$ performs:

    \begin{enumerate}
        \item Sample a fresh opening nonce $r_i$.
        \item Compute the commitment:
        $$c_i = H(tx_i \| r_i \| \idcom_i \| s)$$
        \item Authenticate the commitment:
        $$
        \sigma_i^C = \Sign\bigl(sk_i^{\mathsf{auth}},\; (\texttt{"commit"}, \idcom_i, c_i, s)\bigr)
        $$
        \item Broadcast $(\idcom_i, c_i, s, \sigma_i^C)$ to validators.
    \end{enumerate}

    A validator $v_j$ receiving the commitment verifies:
    \begin{enumerate}
        \item $\idcom_i$ is registered and active;
        \item $\sigma_i^C$ verifies under the registered $vk_i^{\mathsf{auth}}$;
        \item the identity has not exceeded the per-slot quota $\ell$.
    \end{enumerate}

    If all checks succeed, the validator signs a receipt
    $$
    \tau_{i,j}^C = \Sign\bigl(sk_j^{\mathsf{val}},\; (\texttt{"commit"}, \idcom_i, c_i, s)\bigr).
    $$
    Once the user or network aggregates $q_c$ such receipts, the commitment becomes admissible and carries certificate $R_i^C$.

    \paragraph{Set-level quota.} Each validator checks the quota only against the commitments it has seen. An identity that sends different commitments to different quorums can therefore collect certificates for up to $\lfloor (2f+1)\ell/(f+1) \rfloor$ commitments, which approaches $2\ell$, because each honest validator signs at most $\ell$ receipts for it and each certificate needs $f+1$ honest receipts. The quota is therefore also a block validity rule: a block whose admissible set contains more than $\ell$ commitments of one identity is invalid, and $\ell+1$ commitment signatures of one identity for the same slot are slashable evidence against that identity.

    \begin{algorithm}
    \caption{Commit Phase (User $u_i$)}
    \label{alg:commit}
    \begin{algorithmic}
        \REQUIRE Transaction $tx_i$, slot number $s$, auth key $sk_i^{\mathsf{auth}}$
        \STATE Sample fresh nonce $r_i$
        \STATE $c_i \gets H(tx_i \| r_i \| \idcom_i \| s)$
        \STATE $\sigma_i^C \gets \Sign(sk_i^{\mathsf{auth}}, (\texttt{"commit"}, \idcom_i, c_i, s))$
        \STATE Broadcast $(\idcom_i, c_i, s, \sigma_i^C)$ and collect $R_i^C$
    \end{algorithmic}
    \end{algorithm}

    The commit phase ends at cutoff $\Delta_c$. Only commitments with valid receipt bundles $R_i^C$ collected before this cutoff belong to the slot's admissible set.

    \subsection{Phase 2: Order}
    \label{sec:order-phase}

    After the commitment cutoff, the producer constructs the admissible commitment set
    $$
    \mathcal{C}_s^* = \{(\idcom_i, c_i, R_i^C)\}
    $$
    containing every commitment certificate valid for slot $s$. The set is canonicalized by sorting on $(\idcom_i, c_i)$.

    The ordering then proceeds as follows:

    \begin{enumerate}
        \item Compute the admissible-set root
        $$
        \mathcal{C}_{s,\mathsf{root}} = \mathsf{MerkleRoot}(\mathsf{Sort}(\mathcal{C}_s^*)).
        $$
        \item Sign $(\texttt{"lock"}, \mathcal{C}_{s,\mathsf{root}}, s)$ and broadcast it to the validators before the root-lock deadline $\Delta_c + \delta_{\mathsf{lock}}$.
        \item Evaluate the VDF on
        $$
        x_s = H(\mathsf{prev\_block\_hash} \| \mathcal{C}_{s,\mathsf{root}} \| s)
        $$
        to obtain $(\mathsf{seed}_s, \pi_{\VDF,s}) = \VDF.\mathsf{Eval}(\mathsf{pp}, x_s)$.
        \item Derive the execution permutation
        $$
        \sigma_s = \pi_{\mathsf{seed}_s} \in S_m
        $$
        via a seeded Fisher--Yates shuffle over the $m = |\mathcal{C}_s^*|$ admissible commitments.
        \item Publish $(\mathcal{C}_{s,\mathsf{root}}, \mathsf{seed}_s, \pi_{\VDF,s}, \sigma_s)$ in the block proposal together with the certified admissible set.
    \end{enumerate}

    \begin{algorithm}
    \caption{Order Phase (Block Producer)}
    \label{alg:order}
    \begin{algorithmic}
        \REQUIRE Certified admissible set $\mathcal{C}_s^*$, previous block hash $h_{\mathsf{prev}}$
        \STATE $\mathcal{C}_{s,\mathsf{root}} \gets \mathsf{MerkleRoot}(\mathsf{Sort}(\mathcal{C}_s^*))$
        \STATE Sign and broadcast $(\texttt{"lock"}, \mathcal{C}_{s,\mathsf{root}}, s)$ before $\Delta_c + \delta_{\mathsf{lock}}$
        \STATE $x_s \gets H(h_{\mathsf{prev}} \| \mathcal{C}_{s,\mathsf{root}} \| s)$
        \STATE $(\mathsf{seed}_s, \pi_{\VDF,s}) \gets \VDF.\mathsf{Eval}(\mathsf{pp}, x_s)$
        \STATE $\sigma_s \gets \mathsf{FisherYates}(\mathsf{seed}_s, |\mathcal{C}_s^*|)$
        \STATE Publish $(\mathcal{C}_{s,\mathsf{root}}, \mathsf{seed}_s, \pi_{\VDF,s}, \sigma_s)$
    \end{algorithmic}
    \end{algorithm}

    \paragraph{Why VDF, not a simple hash?} If the seed were a plain hash of the admissible-set root, the producer could learn the final permutation immediately once the set is locked and then decide which of its own committed transactions to open. The VDF inserts a mandatory sequential delay between admissible-set lock and permutation availability, forcing commitment decisions to be made before the order is known.

    \paragraph{Root lock.} The delay alone does not stop the producer from choosing the set. After the cutoff it can leave out some of its own certified commitments, evaluate the VDF on several candidate roots in parallel, and publish the root whose permutation it prefers. The root lock removes this choice. Honest validators reject a block whose admissible-set root differs from the signed lock they received before $\Delta_c + \delta_{\mathsf{lock}}$, or for which they received no lock by then, and two signed locks for the same slot are equivocation evidence slashable by $\delta_{\mathsf{prod}} B_{\mathsf{prod}}$. Because $\Delta_v$ exceeds $\delta_{\mathsf{lock}}$ plus the producer's maximum lead in learning the final root (Section~\ref{sec:performance}), no VDF evaluation on a candidate root finishes before the producer must sign one.

    \subsection{Phase 3: Open}
    \label{sec:open-phase}

    After the producer publishes the ordering material, users reveal transaction openings:

    \begin{enumerate}
        \item Broadcast $(tx_i, r_i, \idcom_i, s)$ for each admissible commitment they wish to open.
        \item Validators verify $c_i \stackrel{?}{=} H(tx_i \| r_i \| \idcom_i \| s)$ against the certified admissible commitment.
        \item If valid, validators issue opening receipts
        $$
        \tau_{i,j}^O = \Sign\bigl(sk_j^{\mathsf{val}},\; (\texttt{"open"}, \idcom_i, c_i, s)\bigr),
        $$
        and any opening with at least $q_o$ such receipts obtains certificate $R_i^O$.
        \item The producer executes every transaction with a valid opening certificate in the order induced by $\sigma_s$, skipping unopened commitments and applying the user-side non-opening penalty.
    \end{enumerate}

    \paragraph{Block content rule.} The ordered segment of a block for slot $s$ contains exactly the transactions whose openings carry a valid certificate $R_i^O$ for a commitment in the slot's admissible set, executed in the order induced by $\sigma_s$. A block that executes any other transaction in this segment, or a certified transaction at a position other than its canonical one, is invalid; validators check this from the certificates carried in the block. Without this rule the producer could insert its own uncommitted transactions around admitted ones and the ordering guarantees below would not hold.

    \begin{algorithm}
    \caption{Open Phase (User $u_i$ and Validators)}
    \label{alg:open}
    \begin{algorithmic}
        \REQUIRE Opening $(tx_i, r_i, \idcom_i, s)$, certified commitment $(\idcom_i, c_i, R_i^C)$
        \STATE \textbf{User:} Broadcast $(tx_i, r_i, \idcom_i, s)$
        \STATE \textbf{Validator:} Verify $c_i = H(tx_i \| r_i \| \idcom_i \| s)$
        \STATE \textbf{Validator:} Issue opening receipt and aggregate $R_i^O$
        \STATE \textbf{Producer:} Execute every certified opening in order $\sigma_s$
    \end{algorithmic}
    \end{algorithm}

    \paragraph{Non-opening penalties.} If an admissible commitment does not receive a valid opening certificate $R_i^O$ before the opening cutoff, the corresponding identity is slashed by a fraction $\delta_{\mathsf{user}}$ of its bond for that unopened commitment. Selective non-opening therefore carries an explicit cost. The penalty applies only if the ordering material for slot $s$ was published before the opening phase began; otherwise the commitments of the slot lapse without penalty.

    \paragraph{Omission proofs.} If a block omits an admissible commitment or an executable opening, any participant may present an omission proof consisting of the relevant certificate bundle(s): for commit omission, $(\idcom_i, c_i, s, R_i^C)$; for execution omission, $(\idcom_i, c_i, s, R_i^C, tx_i, r_i, R_i^O)$. Validators verify these objects directly and reject the block if the omission proof is valid. A certificate supports an omission proof only if it reached the validators before the cutoff of its phase ($\Delta_c$ for $R_i^C$, the opening cutoff for $R_i^O$). Honest validators record the certificates they receive before each cutoff and forward them, so under a network delay bound $\delta_{\mathsf{net}}$ the producer holds each such certificate within $\delta_{\mathsf{net}}$ of the cutoff; the protocol requires $\delta_{\mathsf{net}} \le \delta_{\mathsf{lock}}$ and $\delta_{\mathsf{net}} \le \Delta_{\mathsf{margin}}$. An honest validator accepts an omission proof only for a certificate it recorded before the cutoff, only if the identity has fewer than $\ell$ commitments in the admissible set, and only until it votes on the block of the slot. A user who withholds a certificate past the cutoff therefore cannot invalidate a block. A proposal rejected by $2f+1$ validator votes that each cite a valid omission proof is slashable by $\delta_{\mathsf{prod}} B_{\mathsf{prod}}$, with the signed proposal and the votes as evidence. An omission proof that reaches honest validators only after the slot is finalized has no effect: the finalized block stands and the producer is not slashed on its basis.

    \subsection{Timing and Slot Structure}

    Within a single slot of duration $\Delta$, the protocol allocates time to commitment admission, delayed randomness, and opening:

    \begin{center}
    \begin{tabularx}{\columnwidth}{@{}llX@{}}
    \toprule
    \textbf{Phase} & \textbf{Duration} & \textbf{Action} \\
    \midrule
    Commit & $[0, \Delta_c)$ & Users broadcast commitments and collect $R_i^C$ \\
    VDF computation & $[\Delta_c, \Delta_c + \Delta_v)$ & Producer locks the root by $\Delta_c + \delta_{\mathsf{lock}}$ and evaluates delayed randomness on $\mathcal{C}_s^*$ \\
    Open & $[\Delta_c + \Delta_v, \Delta)$ & Users reveal and collect $R_i^O$ \\
    \bottomrule
    \end{tabularx}
    \end{center}

    The timing constraint is
    $$
    \Delta = \Delta_c + \Delta_v + \Delta_o + \Delta_{\mathsf{margin}},
    $$
    where $\Delta_o$ denotes the minimum opening/verification budget and $\Delta_v$ must exceed $\delta_{\mathsf{lock}}$ plus the maximum lead by which a producer can know the final admissible-set root before commit cutoff. This separates the \emph{security requirement} on delayed randomness from any one implementation's raw CPU benchmark.

    \paragraph{Single-slot execution.} Unlike multi-block commit-reveal constructions, MEV-ACE can complete all critical steps in one slot when the chain's slot budget is sized to accommodate $\Delta_c$, $\Delta_v$, and $\Delta_o$. On faster chains, the same logic can be preserved by moving delayed randomness or receipt aggregation to a pipelined implementation.

    % ============================================================
    \section{Security Analysis}
    \label{sec:security}
    % ============================================================

    We analyze the three properties from Section~\ref{sec:model} and then derive the producer's incentive constraints.

    \subsection{Order-Unpredictability}

    \begin{theorem}[Order-Unpredictability]
    \label{thm:order-unpredict}
    Assume that: (i) the admissible commitment set $\mathcal{C}_s^*$ is fixed by the root lock before any VDF evaluation on a candidate root can finish, that is, $\Delta_v$ exceeds $\delta_{\mathsf{lock}}$ plus the producer's maximum lead in learning the final root; (ii) the VDF satisfies the sequential computation assumption; and (iii) $H$ is modeled as a random oracle. Then MEV-ACE satisfies order-unpredictability (Definition~\ref{def:order-unpredict}).
    \end{theorem}

    \begin{proof}
    Let $c_i$ be any admissible commitment in slot $s$ with $m = |\mathcal{C}_s^*|$. The protocol publishes a permutation $\sigma_s$ over the $m$ commitment indices before the open phase begins:
    \begin{align*}
    \sigma_s &= \pi_{\mathsf{seed}_s}, \\
    \mathsf{seed}_s
    &= \VDF.\mathsf{Eval}\bigl(
    \mathsf{pp},
    H(\mathsf{prev\_block\_hash} \| \mathcal{C}_{s,\mathsf{root}} \| s)
    \bigr).
    \end{align*}
    By assumption, $\mathcal{C}_{s,\mathsf{root}}$ is not fixed until the root lock, and no seed for any candidate root is available before the lock. The producer therefore cannot choose among candidate sets by their seeds, and a second signed root is slashable equivocation. Hence any strategy that conditions commit-phase decisions on the realized index $\sigma_s(i)$ cannot use the actual final seed.

    Once $\mathcal{C}_s^*$ is fixed, the seed is derived from a random-oracle hash and a sequentially evaluated VDF output. Therefore, before $\mathsf{seed}_s$ is released, $\Adv$ gains no non-negligible advantage over guessing the index of $c_i$ in the Fisher--Yates permutation. That index is uniformly distributed over $[m]$, so
    $$\Pr[\Adv \text{ predicts } \sigma_s(i)] = \frac{1}{m} + \negl(\lambda).$$
    \end{proof}

    \subsection{Commitment Authenticity}

    \begin{theorem}[Commitment Authenticity]
    \label{thm:commit-unforge}
    Under EUF-CMA security of the authentication signature scheme and collision resistance of $H$, MEV-ACE satisfies commitment authenticity (Definition~\ref{def:commit-unforge}).
    \end{theorem}

    \begin{proof}
    For a commitment to be admissible, validators must verify a valid signature
    $$
    \sigma_i^C = \Sign(sk_i^{\mathsf{auth}}, (\texttt{"commit"}, \idcom_i, c_i, s))
    $$
    under the registered verification key $vk_i^{\mathsf{auth}}$ whose hash equals $\idcom_i$. Therefore an adversary that produces an admissible commitment for identity $\idcom_i$ without controlling $sk_i^{\mathsf{auth}}$ directly yields an EUF-CMA forgery against the signature scheme.

    Collision resistance of $H$ ensures that, after admission, the user cannot later open the same commitment to two distinct tuples $(tx_i, r_i)$ and $(tx_i', r_i')$ with the same $\idcom_i$ and slot number. Thus the commitment is both authenticated at admission time and binding at opening time. The combined failure probability is negligible.

    Authenticity concerns which identity signed a commitment, not how many commitments it may have admitted. An identity that sends different commitments to different quorums can collect certificates for more than $\ell$ of them without any forgery; that count is bounded by the set-level quota rule of Section~\ref{sec:commit-phase}, a block validity rule outside this reduction.
    \end{proof}

    \subsection{Sybil Resistance via Identity Cost}
    \label{sec:sybil-cost}

    \begin{lemma}[Bounded Stuffing Cost]
    \label{lem:sybil}
    Let $\ell$ be the maximum number of commitments allowed per identity per slot. To inject $k$ additional admissible commitments in a slot beyond its currently registered capacity, an adversary must control at least $\lceil k / \ell \rceil$ additional active identities and therefore lock at least $\lceil k / \ell \rceil \cdot D_{\mathsf{stake}}$ of slashable capital.
    \end{lemma}

    \begin{proof}
    By the set-level quota rule (Section~\ref{sec:commit-phase}), each active identity can contribute at most $\ell$ admissible commitments to a slot. Therefore $k$ extra admissible commitments require at least $\lceil k / \ell \rceil$ extra identities. Each such identity must be registered with an active bond $D_{\mathsf{stake}}$, so the minimum additional slashable capital is $\lceil k / \ell \rceil \cdot D_{\mathsf{stake}}$.
    \end{proof}

    \begin{corollary}
    If $u$ of those extra commitments are later left unopened, the adversary incurs additional slashing of at least $u \cdot \delta_{\mathsf{user}} \cdot D_{\mathsf{stake}}$.
    \end{corollary}

    Identity creation remains possible, but commitment stuffing is no longer nearly free: it consumes additional bonded capacity and, under selective non-opening, incurs explicit slashing.

    \subsection{Accountable Inclusion}

    \begin{theorem}[Accountable Inclusion]
    \label{thm:censor-resist}
    Assume that: (i) $q_c = q_o = 2f+1$; (ii) $R_i^C$ and $R_i^O$ reach every honest validator before the cutoff of their phase, and messages between honest parties arrive within $\delta_{\mathsf{net}}$; and (iii) the underlying BFT protocol has more than $2/3$ honest stake. Then MEV-ACE satisfies accountable inclusion (Definition~\ref{def:censor-resist}).
    \end{theorem}

    \begin{proof}
    Let transaction $tx_i$ have an admissible commitment certificate $R_i^C$ and an executable opening certificate $R_i^O$. Because both thresholds are $2f+1$ in a $3f+1$ validator set, each certificate includes signatures from at least $f+1$ honest validators.

    Suppose the producer omits the commitment entirely. Then any observer can present $(\idcom_i, c_i, s, R_i^C)$; all validators can verify the certificate independently and conclude that the commitment belonged to the admissible set. A block whose admissible-set root excludes that commitment is invalid.

    Suppose the producer includes the commitment but omits execution after a valid opening. Then any observer can present $(\idcom_i, c_i, s, R_i^C, tx_i, r_i, R_i^O)$. Validators verify the commitment hash, the opening certificate, and the canonical position of the transaction under $\sigma_s$. A block that does not execute this transaction is invalid.

    By assumption~(ii), every honest validator recorded both certificates before their cutoffs and accepts the omission proof before it votes, so under the assumed BFT quorum honest validators reject any invalid block. Therefore every transaction with both valid certificates appears in every finalized block for the slot.
    \end{proof}

    \subsection{Producer Incentives}

    \begin{theorem}[Honest Execution as Producer Best Response]
    \label{thm:best-response}
    Let $B_{\mathsf{prod}}$ denote the producer bond. Suppose the following inequalities hold for every slot:
    \begin{enumerate}
        \item \textbf{Invalid-block deviation bound:}
        $$
        \delta_{\mathsf{prod}} \cdot B_{\mathsf{prod}} > G_{\mathsf{invalid}},
        $$
        where $G_{\mathsf{invalid}}$ is the maximum one-slot gain from reordering or omitting mandatory transactions.
        \item \textbf{Selective non-opening bound:}
        $$
        \delta_{\mathsf{user}} \cdot D_{\mathsf{stake}} > G_{\mathsf{open}},
        $$
        where $G_{\mathsf{open}}$ is the maximum gain from withholding one admissible producer-controlled opening after observing the realized order.
        \item \textbf{Stuffing-cost bound:} for every $k \ge 1$,
        $$
        \Bigl\lceil \frac{k}{\ell} \Bigr\rceil \rho\, D_{\mathsf{stake}} > G_{\mathsf{stuff}}(k),
        $$
        where $G_{\mathsf{stuff}}(k)$ is the incremental expected gain per slot from injecting $k$ additional admissible commitments and $\rho$ is the per-slot cost of capital. A bond that is not slashed is returned and can back commitments in every later slot, so its cost to the producer is the carrying cost $\rho D_{\mathsf{stake}}$ per slot, not $D_{\mathsf{stake}}$.
    \end{enumerate}
    Then, after the admissible commitment set is fixed, honest execution of the certified permutation is a best response for the producer.
    \end{theorem}

    \begin{proof}
    Consider the producer's possible deviations.

    \textbf{Deviation 1: Reorder or omit mandatory transactions.} By Theorem~\ref{thm:censor-resist}, omission proofs and deterministic recomputation of $\sigma_s$ make the block invalid. The producer loses at least $\delta_{\mathsf{prod}} B_{\mathsf{prod}}$, which exceeds the gain cap $G_{\mathsf{invalid}}$ by assumption.

    \textbf{Deviation 2: Publish a fake VDF output.} VDF soundness makes successful forgery negligible, and an invalid VDF proof again triggers the invalid-block penalty. Expected gain is negative by the first inequality.

    \textbf{Deviation 3: Inject additional admissible commitments.} By Lemma~\ref{lem:sybil}, $k$ extra commitments require at least $\lceil k/\ell \rceil$ additional bonded identities. By the stuffing-cost inequality, the per-slot carrying cost of these bonds exceeds the incremental expected gain $G_{\mathsf{stuff}}(k)$.

    \textbf{Deviation 4: Selectively withhold producer-controlled openings.} Each unopened admitted commitment loses at least $\delta_{\mathsf{user}} D_{\mathsf{stake}}$, which by assumption exceeds the maximum gain $G_{\mathsf{open}}$ from suppressing that opening.

    Since every profitable deviation channel is closed either cryptographically or economically, honest execution is a best response for the producer once the admissible set is fixed.
    \end{proof}

    \subsection{Security Boundaries and Limitations}
    \label{sec:limitations}

    The guarantees are subject to the following limitations:

    \begin{itemize}
        \item \textbf{Conditional transactions.} The non-opening penalty applies only to commitments that are not opened. A producer-controlled identity can commit transactions whose effect depends on chain state (for example, ``buy only if the pool price exceeds $X$, otherwise revert'') and open all of them, so it pays no penalty; once $\sigma_s$ is public, the ones that land around a victim's transaction take effect. Order unpredictability (Definition~\ref{def:order-unpredict}) constrains prediction before the commit cutoff and does not cover this adaptive behavior, so MEV-ACE does not exclude such probabilistic sandwiching by its own means.
        \item \textbf{The guarantees apply only after receipt thresholds are met.} If a user fails to collect $R_i^C$ or $R_i^O$ before the relevant cutoff, the transaction never becomes mandatory for that slot.

        \item \textbf{Information-based and cross-domain MEV remain out of scope.} Back-running on public state changes and arbitrage across chains still require orthogonal mechanisms such as batch auctions or delayed disclosure~\cite{budish2015high}.

        \item \textbf{The protocol depends on timely dissemination of certificates.} A certificate that reaches some honest validators only after the cutoff of its phase does not support an omission proof, and an omission proof that arrives after finalization does not affect the finalized block or slash the producer.

        \item \textbf{Skipped and empty slots.} Accountable inclusion constrains the blocks that are finalized for a slot. A producer that publishes no ordering material, or proposes no block, for slot $s$ excludes the whole batch of that slot, because commitments are bound to $s$ and must be resubmitted in a later slot. The producer forgoes the slot reward, users are not slashed for non-opening, and repeated skipping is handled only by the liveness and producer-rotation rules of the consensus layer. A block with an empty or partial admissible set is not a way around omission proofs, because each omitted certificate recorded by the honest validators invalidates it.

        \item \textbf{VDF hardware asymmetry remains a deployment risk.} If one party can compute the VDF materially faster than the calibrated bound, delayed randomness may no longer provide the intended uncertainty window.

        \item \textbf{Economic guarantees are parameter-sensitive.} If $D_{\mathsf{stake}}$, $\delta_{\mathsf{user}}$, or $\delta_{\mathsf{prod}}$ are set below realistic one-slot profit opportunities, selective non-opening or stuffing can become rational again.
    \end{itemize}

    % ============================================================
    \section{Performance Analysis}
    \label{sec:performance}
    % ============================================================

    We treat cryptographic delay parameters, wall-clock runtime, and end-to-end slot overhead separately.

    \subsection{Computation and Communication Costs}

    Let $m$ be the number of admissible commitments in a slot and $n = 3f+1$ the validator count. MEV-ACE adds four main cost components:

    \begin{itemize}
        \item \textbf{User-side work:} each transaction adds one commitment hash, one opening nonce, and one authentication signature.
        \item \textbf{Validator-side work:} commit and open admission require $O(m)$ signature/hash checks over the slot.
        \item \textbf{Canonicalization:} constructing the admissible-set root requires an $O(m \log m)$ sort and $O(m)$ hashes.
        \item \textbf{Delayed randomness:} VDF evaluation is deployment-specific, while verification remains $O(\log T)$ or better depending on the construction.
    \end{itemize}

    Communication overhead arises from two additional user messages per transaction (commit and open) plus validator receipt certificates, so MEV-ACE is bandwidth-sensitive unless receipt signatures are batched, aggregated, or compressed in the execution layer.

    \subsection{Latency Budget}

    The slot budget should be reasoned about as
    $$
    \Delta_{\mathsf{slot}} =
    \Delta_c + \Delta_v + \Delta_o + \Delta_{\mathsf{margin}},
    $$
    where:
    \begin{itemize}
        \item $\Delta_c$ is the time for users to obtain commit certificates $R_i^C$;
        \item $\Delta_v$ is the wall-clock delay introduced by the randomness mechanism;
        \item $\Delta_o$ is the time for users to obtain opening certificates $R_i^O$;
        \item $\Delta_{\mathsf{margin}}$ absorbs network jitter and consensus scheduling slack.
    \end{itemize}

    The key design rule is that $\Delta_v$ must exceed $\delta_{\mathsf{lock}}$ plus the producer's maximum timing advantage in learning the final admissible-set root; it is set by this security requirement, not by a CPU microbenchmark. A deployment can therefore be single-slot only if its chain-level slot duration is large enough to hold all three phases with margin.

    \subsection{Throughput Implications}

    MEV-ACE does not impose a fixed percentage throughput penalty independent of deployment. Its throughput impact depends on three variables:
    \begin{enumerate}
        \item the validator count and signature scheme used for receipts;
        \item whether receipt evidence is carried as individual signatures, aggregates, or committee certificates;
        \item the chain's networking budget for doubling user-plane messages from one to two phases.
    \end{enumerate}

    Consequently, a credible throughput claim requires an implementation benchmark with a concrete validator topology, signature scheme, aggregation strategy, and slot scheduler. We therefore report no fixed overhead or TPS figure.

    \subsection{Qualitative Comparison}

    Relative to plain commit-reveal, MEV-ACE adds authenticated admissibility and omission proofs. Relative to threshold-encryption systems, it avoids a decryption committee but still requires a carefully calibrated delay budget. Relative to PBS, it targets the source of ordering discretion instead of reallocating extracted value between builders and proposers.

    % ============================================================
    \section{Discussion}
    \label{sec:discussion}
    % ============================================================

    \subsection{Role of Identity Binding}

    Identity alone does not solve MEV. The property MEV-ACE relies on is \emph{authenticated admissibility backed by slashable economic identities}. Fair ordering without authenticated admission still lets the producer manufacture demand; authenticated admission without delayed randomness still lets the producer condition reveals on realized order. MEV-ACE pairs these two controls and adds receipt-backed inclusion so omission becomes globally auditable.

    \subsection{Generalization Beyond ACE-GF}

    While MEV-ACE is presented in the context of ACE-GF, the design generalizes to any system that provides:
    \begin{enumerate}
        \item A registered verification key or equivalent public credential for protocol authentication.
        \item A slashable participation bond or other enforceable cost for expanding identity capacity.
        \item A way to dedicate protocol-specific signing material so that ordering authentication is isolated from other applications.
    \end{enumerate}

    ACE-GF supplies item~(3) directly and can instantiate the authentication key with post-quantum signatures. The security proofs, however, rely only on registered keys, receipt thresholds, and bonded participation, not on any ACE-GF-specific wallet workflow.

    \subsection{Parameter Selection Guidelines}

    \begin{itemize}
        \item \textbf{Receipt thresholds $q_c, q_o$}: Set to exactly $2f+1$, so that every certificate contains signatures from at least $f+1$ honest validators and the $f$ Byzantine validators cannot block a certificate by withholding receipts.
        \item \textbf{Delay budget $\Delta_v$}: Choose $\Delta_v$ to exceed $\delta_{\mathsf{lock}}$ plus the producer's maximum timing advantage in learning the final admissible-set root, with additional margin for hardware asymmetry and network jitter.
        \item \textbf{Stake deposit $D_{\mathsf{stake}}$}: Size the user bond so that $\delta_{\mathsf{user}} D_{\mathsf{stake}}$ exceeds the maximum gain from withholding one admitted opening, and the carrying cost $\lceil k/\ell \rceil \rho D_{\mathsf{stake}}$ exceeds the per-slot gain from stuffing $k$ extra commitments.
        \item \textbf{Producer bond $B_{\mathsf{prod}}$}: Size the producer-side slash so that $\delta_{\mathsf{prod}} B_{\mathsf{prod}}$ exceeds the maximum one-slot benefit from invalid ordering or omission.
        \item \textbf{Rate limit $\ell$}: Choose $\ell$ based on expected user activity. A larger $\ell$ improves UX for high-frequency users but weakens the capital cost of commitment stuffing.
    \end{itemize}

    % ============================================================
    \section{Related Work}
    \label{sec:related}
    % ============================================================

    \paragraph{MEV taxonomy and measurement.} Daian et al.~\cite{daian2020flash} introduced the concept of miner extractable value and documented front-running on Ethereum. Flashbots later published MEV-Explore as a practical measurement resource for extracted MEV~\cite{flashbots-explore}. Qin et al.~\cite{qin2022quantifying} quantified extractable value across major DeFi activity classes.

    \paragraph{Commit-reveal and encryption-based approaches.} Canidio and Danos~\cite{canidio2023commit} analyze commit-reveal schemes against front-running. Shutter Network~\cite{shutter2023} uses threshold encryption to hide transaction contents from sequencers. These approaches reduce early information leakage, but by themselves they do not authenticate admissible commitments via bonded identities or provide portable omission proofs.

    \paragraph{Fair ordering services.} Kelkar et al.~\cite{kelkar2020order} formalized order-fairness and proposed Aequitas, which enforces it inside the consensus protocol from the receive orders of the validators; MEV-ACE instead fixes order by delayed randomness over a locked commitment set and does not use receive order. Chainlink Fair Sequencing Services (FSS)~\cite{chainlink-fss} decentralizes ordering via oracle networks, an external trust assumption that MEV-ACE avoids.

    \paragraph{Proposer-builder separation.} Flashbots MEV-Boost~\cite{flashbots-mev-boost} separates block building from proposing in Ethereum. The builder still extracts MEV and the proposer receives a share.

    \paragraph{VDF-based protocols.} Boneh et al.~\cite{boneh2018vdf} formalized VDFs. Wesolowski~\cite{wesolowski2019efficient} provided efficient VDF constructions. In MEV-ACE, delayed randomness is one component alongside authenticated admission and receipt-backed inclusion.

    \paragraph{Identity and Sybil resistance.} Proof-of-personhood protocols~\cite{borge2017proof} address Sybil resistance through identity verification. ACE-GF~\cite{acegf} provides a cryptographic identity framework with context-isolated key derivation. MEV-ACE differs from both strands by using registered, slashable protocol identities to authenticate commitments and openings in the block-building path.

    % ============================================================
    \section{Conclusion}
    \label{sec:conclusion}
    % ============================================================

    We have presented MEV-ACE, an identity-authenticated fair-ordering protocol for proposer-controlled MEV mitigation. Our central claim is that, once commitments and openings cross threshold receipt gates, the producer no longer retains unilateral discretion to reorder or omit them in a finalized block without either violating verifiable protocol rules or paying explicitly modeled economic penalties. This result follows from the combination of registered authentication keys, receipt-backed admissibility, delayed randomness, and slashable participation bonds.

    MEV-ACE provides order-unpredictability, commitment authenticity, and accountable inclusion only under standard cryptographic assumptions, timely dissemination of omission proofs, and correctly sized economic parameters. ACE-GF serves as the deterministic identity/key-derivation substrate and preserves compatibility with post-quantum signatures such as ML-DSA-44.

    \paragraph{Future work.} The next three priorities are: (1)~benchmarking concrete receipt and VDF implementations under realistic validator topologies; (2)~formalizing omission-proof handling inside a full BFT state machine; and (3)~studying how this framework composes with batch auctions or encrypted execution for information-based MEV.

    \bibliographystyle{plain}

\end{document}